\documentclass{lmcs}
\pdfoutput=1

\usepackage[utf8]{inputenc}

\usepackage{lastpage}
\lmcsdoi{17}{4}{17}
\lmcsheading{}{\pageref{LastPage}}{}{}%
{May~17,~2021}{Dec.~13,~2021}{}


\usepackage{algpseudocode, algorithm}

\DeclareMathOperator{\ex}{\mathbb{E}}
\DeclareMathOperator{\R}{\mathbb{R}}
\DeclareMathOperator{\ftwo}{\mathbb{F}_2}

\DeclareMathOperator{\calC}{\mathcal{C}}
\DeclareMathOperator{\ralg}{\mathsf{CSP-SAT}}

\begin{document}

\title{A Separator Theorem for Hypergraphs and a CSP-SAT Algorithm}

\author[M.~Kouck{\'y}]{Michal Kouck{\'y}}	
\address{Computer Science Institute of Charles University, Prague, Czech Republic}	
\email{koucky@iuuk.mff.cuni.cz}  

\author[V.~R{\"o}dl]{Vojt{\v e}ch R{\"o}dl}
\address{Emory University, Atlanta, USA}
\email{vrodl@emory.edu}

\author[N.~Talebanfard]{Navid Talebanfard}
\address{Institute of Mathematics of the Czech Academy of Sciences, Prague, Czech Republic}
\email{talebanfard@math.cas.cz}




\begin{abstract}
We show that for every $r \ge 2$ there exists $\epsilon_r > 0$ such that any $r$-uniform hypergraph with $m$ edges and maximum vertex degree $o(\sqrt{m})$ contains a set of at most $(\frac{1}{2} - \epsilon_r)m$ edges the removal of which breaks the hypergraph into connected components with at most $m/2$ edges. We use this to give an algorithm running in time $d^{(1 - \epsilon_r)m}$ that decides satisfiability of $m$-variable $(d, k)$-CSPs in which every variable appears in at most $r$ constraints, where $\epsilon_r$ depends only on $r$ and $k\in o(\sqrt{m})$. Furthermore our algorithm solves the corresponding \#CSP-SAT and Max-CSP-SAT of these CSPs. We also show that CNF representations of unsatisfiable $(2, k)$-CSPs with variable frequency $r$ can be refuted in tree-like resolution in size $2^{(1 - \epsilon_r)m}$. Furthermore for Tseitin formulas on graphs with degree at most $k$ (which are $(2, k)$-CSPs) we give a deterministic algorithm finding such a refutation.
\end{abstract}

\maketitle

\section{Introduction}

The $(d, k)$-SAT problem which naturally generalizes $k$-SAT is the problem of deciding whether a system of constraints on $m$ variables from an alphabet of size $d$, where each constraint is on at most $k$ variables, can be satisfied. We will call such a system of constraints a $(d, k)$-CSP and we will assume that when given as input it is represented by the set of truth tables of its constraints. Therefore the satisfiability of a $(d, k)$-CSP $\Psi$ can be checked by exhaustive search in time $|\Psi|^{O(1)}d^m$. Therefore looking for exponential time algorithms beating this trivial running time is a natural direction. 

For the usual $k$-SAT problem, where $d = 2$, there is a plethora of such algorithms (see e.g.~\cite{DBLP:journals/cjtcs/PaturiPZ99,DBLP:journals/jacm/PaturiPSZ05,DBLP:journals/tcs/DantsinGHKKPRS02}). When the CSP encodes a certain structured problem we can also find improved algorithms. The notable example here is the graph $d$-coloring problem which is a special case of $(d, 2)$-SAT and which can be solved in time $O(2^m)$~\cite{DBLP:journals/siamcomp/BjorklundHK09}. More generally $(d, 2)$-SAT also admits non-trivial algorithms~\cite{DBLP:journals/jal/BeigelE05}.

For the general $(d, k)$-SAT we are interested in finding algorithms running in time $d^{(1 - \epsilon)m}$ for some $\epsilon > 0$. We call the parameter $\epsilon$ the {\it savings} of the algorithm,  and we would like these savings to be as large as possible. Note that any $k$-SAT algorithm can be easily converted to a $(d, k)$-SAT algorithm. For each of the original variables, introduce $\log d$ boolean variables representing the original value in binary, and then express each constraint as a $k \log d$-CNF. The conjunction of these CNFs is satisfiable if and only if the original CSP is satisfiable. Assuming that we can solve $k$-SAT in time $2^{(1 - \epsilon_k)m}$, this yields an algorithm running in time $d^{(1 - \epsilon_{k \log d})m}$. That is any non-trivial savings for $k$-SAT yields non-trivial savings for $(d, k)$-SAT. However these savings deteriorate as $d$ grows. This turns out to be the case also for algorithms which are directly designed to solve $(d, k)$-SAT. Sch{\"o}ning's seminal algorithm runs in time $O((\frac{d(k - 1)}{k})^m)$~\cite{DBLP:conf/focs/Schoning99}. Similarly a generalization of PPSZ analyzed by Hertli et al.~\cite{DBLP:conf/cp/HertliHMMSS16} has the same shortcoming. The central question is then whether it is possible to obtain savings independent of the domain size $d$.

Let us define $$\sigma_{d,k} := \sup\{\delta: (d,k)\text{-SAT can be solved in time } O(d^{(1 - \delta) m})\}.$$ The argument above gives $\sigma_{d, k} \ge \sigma_{2, k \log d}$. Furthermore Traxler~\cite{DBLP:conf/iwpec/Traxler08} shows that for all $d$, $\sigma_{2, k} \ge \sigma_{d, k}$. Therefore it follows that under Strong Exponential Time Hypothesis, for all $d$, $\lim_{k \rightarrow \infty} \sigma_{d, k} = 0$. Our central question can be rephrased as follows. 

\begin{qu}
Is it the case that for every $k$, $\lim_{d \rightarrow \infty} \sigma_{d, k} > 0$?
\end{qu}

Currently, we are unable to answer this question. However we show that if each variable appears in a small number of constraints then it is possible to decide satisfiability with savings independent of $d$. Note that this restriction does not influence the NP-hardness of $(d, k)$-SAT; in fact even if each variable appears in at most 2 constraints, it is easy to see that the problem remains NP-hard. Our result can be considered as an extension of a result of Wahlstr{\"o}m~\cite{DBLP:conf/esa/Wahlstrom05} who gave such an algorithm for CNF-SAT when variables have bounded occurrences. However our argument is entirely different. Our algorithm also solves the counting and MAX versions of this problem. For Boolean CSPs with bounded occurrence such a result was shown by Chen and Santhanam~\cite{DBLP:conf/sat/ChenS15}.

\newcounter{dpll-alg}
\setcounter{dpll-alg}{\value{thm}}

\begin{thm}[Main result, informally stated]
There exists an algorithm which decides satisfiability of an $m$ variable $(d, k)$-CSP in which every variable appears in at most $r$ constraints in time $d^{(1 - \epsilon)m}$ where $\epsilon$ depends only on $r$, provided that $k$ does not grow too fast (that is $k \in o(\sqrt{m})$).
\end{thm}

The algorithm follows a simple branching strategy. At every step we find a small set of variables that once given a value, the CSP breaks into disjoint parts each with at most half of the original variables. For every assignment on these variables we exhaustively solve the problem on the resulting smaller instances. If this set contains strictly less than half of the variables then we obtain savings. To prove that such a small set of variables exists, we associate a natural hypergraph to the CSP. Then we prove a structural result for these hypergraphs: We show that in every $r$-uniform hypergraph on $m$ edges of small vertex degree, there exists a set of significantly less than half hyperedges the removal of which breaks the hypergraph into connected components with at most $m/2$ hyperedges.

\newcounter{hyp-sep}
\setcounter{hyp-sep}{\value{thm}}

\begin{thm}[Hypergraph separator theorem, informally stated]
Let $H = (V, E)$ be an $r$-uniform hypergraph on $m$ hyperedges and maximum vertex degree $k$. Provided that $k$ does not grow too fast as a function of $m$, there exists a set of $(1 - \epsilon)m/2$ hyperedges the removal of which breaks the hypergraph into connected components with at most $m/2$ hyperedges. Furthermore, $\epsilon$ depends only on $r$.
\end{thm}

We find this result somewhat unexpected. To see this consider one particular consequence. Our result implies that there exists a universal constant $\epsilon > 0$ such for {any} positive integer $d$, any graph with maximum degree at most $d$ with sufficiently many edges can be broken into components with at most half of the edges by removing at most $\frac{1}{2} - \epsilon$ fraction of edges. What is unexpected is this independence between $\epsilon$ and $d$, and furthermore such a result is impossible if we slightly change the definition of balanced separators and require that the number of vertices (instead of edges) would be at most half of the original graph in each component after removing edges. It is easy to see that this is essentially the same as {\it edge-expansion} which is known to be at least $(\frac{1}{2}-O(\frac{1}{\sqrt{d}}))m$ for some $d$-regular graphs \cite{MR947025}. The independence between savings and degree in our result is precisely what we take advantage of in our applications.

\bigskip
\noindent{\bf Non-trivial exponential size proofs of Tseitin formulas.} We provide yet another application of Theorem \ref{thm:sep} which concerns the proof complexity of unsatisfiable $k$-CNF formulas. Recall that $k$-TAUT is the language of $k$-DNF tautologies (or equivalently the language of unsatisfiable $k$-CNFs). Define $$\nu_k := \sup\{\delta : k\text{-TAUT can be solved in non-deterministic time }O(2^{(1 - \delta)m})\}.$$ 

The {\it non-deterministic strong exponential time hypothesis (NSETH)}~\cite{DBLP:conf/innovations/CarmosinoGIMPS16} states that $\lim_{k \rightarrow \infty} \nu_k = 0$. Quantity $\nu_k$ can of course be defined with respect to a specific proof system instead of a general non-deterministic algorithm. In this direction some works verify NSETH for restrictions of the resolution proof system (see~\cite{DBLP:conf/soda/PudlakI00, DBLP:conf/stoc/BeckI13, DBLP:journals/algorithmica/BonacinaT17}). 

But even if NSETH holds we may ask how fast $\nu_k$ approaches zero. Observe that $\nu_k \ge \sigma_{2, k}$ as any $k$-SAT algorithm in particular refutes unsatisfiable $k$-CNFs formulas. The best known lower bound for $\sigma_{2, k}$ is $\Omega(1/k)$~\cite{DBLP:journals/jacm/PaturiPSZ05} and thus $\nu_k \in \Omega(1/k)$. We now raise a very natural question. 

\begin{qu}
Is it the case that $\nu_k \gg 1/k$, that is can we use non-determinism to beat the best known savings of $k$-SAT algorithms? 
\end{qu}

We can think of two ways to make progress towards this question. In the first direction we could try to obtain lower bounds on $\nu_k$ directly by proving upper bounds on the size of refutations of $k$-CNFs. Interestingly (but not surprisingly) the $\Omega(1/k)$ bound can already be achieved by tree-like resolution~\cite{DBLP:journals/ipl/BonacinaT16}, that is every unsatisfiable $k$-CNF formula in $m$ variables can be refuted by a tree-like resolution proof of size $2^{(1 - \Omega(1/k))m}$. 

In the second direction we can consider {\it general enough} families of $k$-CNFs and try to obtain non-trivial refutations for them in as weak as possible a proof system. Here by general enough families we mean families of formulas which naturally contain a wide spectrum of easy to hard instances. An example of such families is the set of $k$-CNFs encoding systems of linear equations over $\ftwo$ each on at most $k$ variables. Indeed such systems can easily be refuted using Gaussian elimination. But for $k = O(1)$, if the underlying system is minimally unsatisfiable, then there are even non-trivial regular resolution refutations of size $2^{m/2 + o(m)}$~\cite{DBLP:journals/cc/Ben-SassonI10}. Such a bound cannot be obtained in tree-like resolution~\cite{DBLP:conf/soda/PudlakI00}. 

However here using our separator theorem we show that if we restrict ourselves to Tseitin formulas over bounded degree graphs (which are also minimally unsatisfiable linear systems and massively studied in proof complexity, see e.g.~\cite{DBLP:journals/jacm/Urquhart87,DBLP:conf/focs/Hastad17}), then we do have tree-like resolution refutations of size $2^{cm}$ for a universal constant $c < 1$. Another unexpected upper bound for Tseitin formulas has recently been observed in~\cite{DBLP:conf/coco/DadushT20} where a cutting plane proof of quasi-polynomial size has been given. Interestingly their proof is also based on a simple branching argument.

\bigskip
\noindent{\bf Deterministic construction of proofs.} We know from a recent breakthrough result of Atserias and M{\"u}ller~\cite{DBLP:journals/jacm/AtseriasM20} that unless P equals NP, given an unsatisfiable formula it is not possible to construct resolution proofs in time polynomial in the size of the of smallest proof of the input formula, nor in quasi-polynomial or subexponential time under plausible stronger assumptions. However it seems that the question of constructing proofs in non-trivial exponential time, natural counterpart to $k$-SAT algorithms, has not been given too much attention. 

More specifically given a proof system $P$ and positive integer $k$ we define $\pi_{P,k}$ as follows. It is the supremum of $\delta$ such that there is a deterministic algorithm that given an unsatisfiable $m$-variable $k$-CNF, it  constructs a $P$-refutation of it in time $O(2^{(1 - \delta) m})$.

\begin{qu}
Is it the case that $\pi_{\text{Res}, k} \in \Omega(1/k)$?
\end{qu}

As mentioned earlier we know that there are even tree-like resolution proofs with such savings. Here we are asking whether we can construct such proofs deterministically. We cannot answer this question yet. But we will show that over Tseitin formulas we can deterministically construct our proofs obtained by the separator theorem with constant savings.

\section{Preliminaries}

A hypergraph is {\it $r$-uniform} if all of its hyperedges have size exactly $r$. A {\it connected component} in a hypergraph $H = (V, E)$ is a maximal subset of vertices $S \subseteq V$ such that for every pair $u, v \in S$ there exists a sequence of edges $e_1, \ldots, e_t$ only on vertices in $S$ with $u \in e_1$ and $v \in e_t$ and $e_i \cap e_{i + 1} \ne \emptyset$ for every $i \in [t - 1]$. For any $S \subseteq V$ we denote by $E(S)$ the set of hyperedges induced on $S$, i.e., all hyperedges from $E$ which are entirely on vertices in $S$. For $S, T \subseteq V$ we define $E(S, T) = \{e \in E : e \cap S \ne \emptyset, e \cap T \ne \emptyset\}$. We will use $\Delta(G)$ to denote the maximum vertex degree of a (hyper)graph $G$, that is the largest number of edges in which a vertex appears.

Given two real numbers $a \ge b$, we write $a \pm b$ to mean the interval $[a - b, a + b]$. To ease reading we write $x = a \pm b$ to mean $x \in [a - b, a + b]$.

We assume the reader is familiar with basic proof complexity, see~\cite{MR2895965, MR3929744} for more background. However we will use the following interpretation of tree-like resolution proofs.

A {\it decision tree} for an unsatisfiable CNF $\Psi$ is a rooted binary tree where the inner nodes are labeled by variables of $\Psi$ and leaves are labeled by clauses of $\Psi$. Edges are labeled by 0 or 1 and should be interpreted as an assignment to the variable they come out of. Therefore a path from the root to a leaf gives an assignment to some of the variables. Furthermore we require that this assignment falsifies the clause at the corresponding leaf. The following theorem is folklore.

\begin{thm}[see, e.g.,~\cite{DBLP:journals/ipl/BeyersdorffGL13}]
\label{thm:dtree}
The size of the smallest tree-like resolution refutation of $\Psi$ is exactly the size of the smallest decision tree for it.
\end{thm}

The following is immediate.

\begin{cor}
\label{cor:dtree}
If there is a decision tree of depth at most $d$ for $\Psi$, then there exists a tree-like resolution refutation of $\Psi$ of size at most $2^d$.
\end{cor}

\section{Hypergraph Separator Theorem}

In this section we present our main technical tool concerning the structure of hypergraphs.

\begin{defi}
Let $H = (V, E)$ be a hypergraph. A balanced separator for $H$ is a set $R \subseteq E$ such that any connected component in $(V, E \setminus R)$ is incident to at most $|E|/2$ edges from $E \setminus R$.
\end{defi}

Note that any set $R \subseteq E$ of size $|E|/2$ is trivially a balanced separator. Therefore the question is whether it is possible to get balanced separators of size strictly less than $|E|/2$. We show that when the hypergraph has bounded vertex degree this is indeed possible. 

We will first show that when the maximum vertex degree of the hypergraph is small enough, small balanced separators do exist. We then show that our bound is quantitatively tight.


We need the following concentration bound. This is a consequence of an inequality due to Alon, Kim and Spencer~\cite{MR1469109} which was used by Dellamonica and R{\"o}dl~\cite{MR4002718}.

\begin{lemC}[\cite{MR1469109,MR4002718}]
\label{lm:concen}
Let $X_1, \ldots, X_n$ be independent Bernoulli variables with $\Pr[X_i = 1] = p$ for all $i \in [n]$. Let $c > 0$ and assume that there is a function $f : \{0, 1\}^n \rightarrow \R$ such that for all $(x_1, \ldots, x_n) \in \{0, 1\}^n$ and every $i \in [n]$, $$|f(x_1, \ldots, x_n) - f(x_1, \ldots, x_{i - 1}, 1 - x_i, x_{i + 1}, \ldots, x_n)| \le c.$$ Then for $\sigma^2 := nc^2p(1-p)$ and any $0 < \alpha < 2\sigma/c$, we have $$\Pr[|X - \ex[X]| \ge \alpha \sigma] \le 2e^{-\alpha^2/4},$$ where $X := f(X_1, \ldots, X_n)$.
\end{lemC}

\newcounter{tmp}
\setcounter{tmp}{\value{thm}}

\setcounter{thm}{\value{hyp-sep}}

\begin{thm}
\label{thm:sep}

Let $r \ge 2$ be fixed and let $H = (V, E)$ be a hypergraph with $m$ hyperedges and maximum vertex degree $k = o(\sqrt{m})$ where each hyperedge has size at most $r$. Then $H$ has a balanced separator of size at most $(\frac{1}{2} - \epsilon_r)m + o(m)$, where $\epsilon_r = (1 - 2^{-1/r})^r$. Furthermore such a balanced separator can be found by a randomized algorithm in expected polynomial time and by a deterministic algorithm in time $m^{O(1)}2^{(1-2\epsilon_r^2)m + o(m)}$.

\end{thm}

For any $1\le k\le m$ the proof of the theorem actually gives a balanced separator of size at most $(\frac{1}{2} - \epsilon_r)m + C_r k \sqrt{m}$, where $C_r>0$ is a universal constant depending only on $r$. We will use this bound later to bound the size of tree-like resolution of Tseitin formulas.

Notice, $\epsilon_r \ge \frac{1}{(2r)^r}$, as for $x\in [0,1]$, $2^{-x} \le 1-\frac{x}{2}$ so $2^{-\frac{1}{r}} \le 1-\frac{1}{2r}$ and thus $(1-2^{-\frac{1}{r}}) \ge 1-(1-\frac{1}{2r}) = \frac{1}{2r}$.

\begin{proof}

We first observe that with a small modification we may assume without loss of generality that the hypergraph is $r$-uniform and contains no isolated vertices. To make the hypergraph $r$-uniform we partition the hyperedges of size less than $r$ into blocks of size $k$. For each of these blocks we introduce $r$ new vertices and we add sufficiently many of them to the hyperedges in the block to make them $r$-uniform. This guarantees that the maximum degree remains at most $k$ and we have added at most $mr/k$ new vertices. We then remove isolated vertices if any exists. We will denote by $n$ the new number of vertices. As there are no isolated vertices $n\le mr$.

The main idea for building the separator is to find a set $S$ of vertices such that $|E(S)| = m/2 \pm o(m)$ and $|E(S, \overline{S})| \le (\frac{1}{2} - \epsilon_r)m \pm o(m)$. Observe that for such $S$, we also have $|E(\overline{S})| \le m/2 + o(m)$. Furthermore note that $E(S)$ and $E(\overline{S})$ are separated by $R := E(S, \overline{S})$, i.e., every connected component in $(V, E \setminus R)$ is entirely contained in $S$ or in $\overline{S}$. Then we arbitrarily select two sets $W_1 \subseteq E(S)$ and $W_2 \subseteq E(\overline{S})$ with $|W_1|, |W_2| \le o(m)$ such that $|E(S) \setminus W_1| \le m/2$ and $|E(\overline{S}) \setminus W_2| \le m/2$ (which is possible by the assumption on $S$). It follows that $R \cup W_1 \cup W_2$ is a balanced separator of size at most $(1 - \epsilon_r)m/2 + o(m)$.

We pick the set $S \subseteq V$ by including each vertex in $S$ independently at random with probability $p := 2^{-1/r}$. Let $ V = [n]$. For each vertex $i \in [n]$, let $X_i$ be the random variable which takes value 1 if vertex $i$ is included in $S$, and it takes value 0 otherwise. Thus we have $S = \{i \in [n] : X_i = 1\}$. We define two functions $f_1$ and $f_2$ as follows: $$f_1(X_1, \ldots, X_n) := |E(\{i \in [n] : X_i = 1\})| = |E(S)|$$ and $$f_2(X_1, \ldots, X_n) := |E(\{i \in [n] : X_i = 1\}, \{i \in [n] : X_i = 0\})| = |E(S, \overline{S})|.$$  Hence observe that $$\ex[f_1(X_1, \ldots, X_n)] = p^rm = \frac{m}{2}$$ and $$\ex[f_2(X_1, \ldots, X_n)] = (1 - p^r - (1 - p)^r)m = \left(\frac{1}{2} - \epsilon_r\right)m.$$

 We would like to apply Lemma \ref{lm:concen} on $f_1$ and $f_2$. Note that since the maximum degree of $H$ is at most $k$, for any $b \in [2]$ and $i \in [n]$, $$|f_b(X_1, \ldots , X_n) - f_b(X_1, \ldots, X_{i - 1}, 1 - X_i, X_{i + 1}, \ldots, X_n)| \le k.$$

 Setting $\alpha = 4$, $c = k$ and $\sigma^2 = nk^2p(1-p)$ we apply Lemma \ref{lm:concen} on $f_1$ and $f_2$ and obtain $$\Pr\left[|f_1(X_1, \ldots, X_n) -\frac{m}{2}| \ge 4k\sqrt{np(1 - p)}\right] \le 2 e^{-4}$$ and $$\Pr\left[|f_2(X_1, \ldots, X_n) - \left(\frac{1}{2} - \epsilon_r\right)m |\ge 4k\sqrt{np(1 - p)}\right] \le 2 e^{-4}.$$ 
 Since $r$ is fixed, $n\le mr$ and $k \in o(\sqrt{m})$ we have $4k\sqrt{np(1 - p)} \le 4k\sqrt{mr} = o(m)$. Therefore with probability at least $1 - 4e^{-4} > 0$, $|E(S)| = m/2 \pm o(m)$ and $|E(S, \overline{S})| \le (\frac{1}{2} - \epsilon_r)m + o(m)$ and hence there exists a choice of $S$ satisfying these properties. As explained earlier by adding at most $o(m)$ edges to $E(S, \overline{S})$ we obtain a balanced separator of size $(\frac{1}{2} - \epsilon_r)m + o(m)$.

It is clear that the above argument also yields a randomized algorithm for finding such a balanced separator. The probability that $S$ satisfies our desired properties is at least a constant and in polynomial time we can verify whether it indeed satisfies both those properties. Thus by repeating the random choice of $S$, in expected polynomial time we find our balanced separator.

The deterministic algorithm exhaustively checks all sets of at most $(\frac{1}{2} - \epsilon_r)m + o(m)$ edges to find a balanced separator. This has running time $m^{O(1)} {m\choose{{(\frac{1}{2}-\epsilon_r)m + o(m)}}} \le 2^{h(1/2 - \epsilon_r)m+o(m)}$, where $h(\cdot)$ is the binary entropy function. Using $h(1/2 - x/2) \le 1 - x^2/2$, we can bound the running time by $m^{O(1)}2^{(1 - 2\epsilon^2_r)m + o(m)}$.
\end{proof}

\setcounter{thm}{\value{tmp}}

We now show that Theorem \ref{thm:sep} is tight. We probabilistically construct a sparse hypergraph in which every sufficiently large set of vertices induces a large number of edges. We then argue that the best balanced separator essentially has to partition the vertices into two parts of a particular size.

\begin{lem}
\label{lm:randomhyp}

For every fixed $\alpha < 1$, $r \ge 2$, and for any $n$ and $k$ such that $k \in o(n)$, there exists an $r$-uniform hypergraph $H = (V, E)$ on $n$ vertices with the following properties:

\begin{enumerate}
\item $|E| = (1 \pm o(1))\frac{nk}{r}$
\item $\Delta(H) = O(k)$
\item For every $S \subseteq V$ with $|S| \ge \alpha n$, $|E(S)| = (1 \pm o(1)){(\frac{|S|}{n})^r}|E|$.
\end{enumerate}
\end{lem}

\begin{proof}
We first sample $G$ from $G^{(r)}(n, q)$ where $q = \frac{nk/r}{{n \choose r}}$, that is we construct $G$ by choosing each possible hyperedge of size $r$ independently with probability $q$. The expected number of edges is $\frac{nk}{r}$. Thus by Chernoff's bound with probability at least $1 - 2^{-\Omega(n)}$, $|E| = (1 \pm o(1))\frac{nk}{r}$.

 Next we bound the number of edges which are incident to vertices of degree at least $2ek$. For every vertex the probability that it has degree at least $t$ is at most ${{{n - 1} \choose {r - 1}}\choose t}q^t$. Note that for $t \ge 2ek$ this probability is at most $2^{-t}$. Therefore the expected number of edges incident with some vertex of degree at least $2ek$ is at most $$n\sum_{t = 2ek}^{\infty} t 2^{-t} =  O(nk/2^{k}).$$

Therefore by Markov's inequality with constant probability the number of edges incident with some vertex of degree at least $2ek$ is at most $O(nk/2^k)$.

Let $S \subseteq V$ be any set of size $\alpha n$. The expected number of edges in $S$ is ${\alpha n \choose r}q$. Let $\delta = k^{-1/3} = o(1)$. By Chernoff's bound 

\begin{eqnarray*}
\Pr\left[|E(S)| \neq (1 \pm \delta){\alpha n \choose r}q\right] &\le& 2\exp\left(\frac{-\delta^2 {\alpha n \choose r}q}{3}\right)\\
&=& 2 \exp\left(\frac{-\delta^2 {\alpha n \choose r}nk}{3r{n \choose r}}\right)\\
&\le& 2\exp\left(\frac{-(1 \pm o(1))\alpha^r nk^{1/3}}{3r}\right), 
\end{eqnarray*}
where in the last inequality we use $\frac{{{\alpha n} \choose r}}{{n \choose r}} = (1 \pm o(1))\alpha^r$ since $r$ is fixed. 
Since there are $n \choose \alpha n$ choices for $S$, the probability that there exists $S \subseteq V$ of size $\alpha n$ with $|E(S)| \neq (1 \pm \delta){\alpha n \choose r}q$ is at most $${n \choose \alpha n} \times 2\exp\left(\frac{-(1 \pm o(1))\alpha^r nk^{1/3}}{3r}\right) = o(1).$$

It follows that there exists an $r$-uniform hypergraph $G$ with $(1 \pm o(1))\frac{n}{r}k$ edges, at most $\gamma n$ of which are incident with some vertex of degree at least $2ek$, where $\gamma = O(k/2^k)$. Furthermore for every $S \subseteq V$ with $|S| = \alpha n$, $|E(S)| = (1 \pm \delta){\alpha n \choose r}q$. We remove at most $\gamma n$ edges to make the maximum vertex degree at most $O(k)$. Call the resulting hypergraph $H = (V, E)$. Note that $|E(H)| = (1 \pm o(1))\frac{n}{r}k$. Let $\delta' = \delta + \gamma$. We have in $H$ for every $S \subseteq V$ with $|S| = \alpha n$, $|E(S)| = (1 \pm \delta'){\alpha n \choose r}q$. A simple averaging argument further gives that for every $S \subseteq V$ with $|S| \ge \alpha n$, $E(S) = (1 \pm \delta'){|S| \choose r}q = (1 \pm o(1))(\frac{|S|}{n})^r|E|$.
\end{proof}

\begin{thm}
\label{thm:sep-lowerbound}
For every fixed $r \ge 2$ and $k, m \rightarrow \infty$ with $k \in o(m)$, there exists an $r$-uniform hypergraph with vertex degree $O(k)$ and $m$ edges such that any balanced separator of $H$ has size at least $(\frac{1}{2} - \epsilon_r)m(1 \pm o(1))$, where $\epsilon_r := (1 - 2^{-1/r})^r$.
\end{thm}

We will show that the hypergraph $H = (V, E)$ given by Lemma \ref{lm:randomhyp} with $\alpha:= 1 - 2^{-1/r}$ satisfies this property. Note that $\epsilon_r = \alpha^r$ and $(1 - \alpha)^r = 1/2$. Let $m = |E|$. The following fact proves the result for the case when the balance separator is a bipartition. As it turns out and which we will show later, this is also the core of the argument for the general case.

\begin{fact}
\label{fc:cutsize}
Let $(A, B)$ be a bipartition of $H$ with $\min\{|A|, |B|\} \ge \alpha n$. Then $|E(A, B)| \ge (1/2 - \epsilon_r)m(1 \pm o(1))$.
\end{fact}

\begin{proof}

Assume $|A| \le |B|$ and $|A| = \gamma n$. Since $\gamma \ge \alpha$, Lemma \ref{lm:randomhyp} guarantees that $|E(A)| = \gamma^rm(1 \pm o(1))$ and $|E(B)| = (1 - \gamma)^rm(1 \pm o(1))$ and hence 

\begin{eqnarray*}
|E(A, B)| &=&(1 - \gamma^r - (1 - \gamma)^r)m(1 \pm o(1))\\
&\ge& (1 - \alpha^r - (1 - \alpha)^r)m(1 \pm o(1))\\
&=& (\frac{1}{2} - \epsilon_r)m(1 \pm o(1)),
\end{eqnarray*}
where the inequality follows since the function $1 - x^r - (1 - x)^r$ is increasing for $x \in (0, \frac{1}{2})$.
\end{proof}

\begin{proof}[Proof of Theorem \ref{thm:sep-lowerbound}]
Let $R \subseteq E$ be a balanced separator in $H$ of minimum size. The removal of edges in $R$ breaks $H$ into two or more connected components each with at most $m/2$ edges. By minimality of $R$ these connected components are induced subgraphs. We group these connected components in two parts $A$ and $B$ such that $||A| - |B||$ is minimized. Assume that $|A| \le |B|$. We have two cases. Either $B$ is connected or it contains more than one connected component. Note that $R \supseteq E(A, B)$. Assume first that $B$ is connected. We have $|A| \ge (\alpha + o(1))n$ since otherwise $|B| > (1 - \alpha + \mu)n$ for some $\mu > 0$ and hence $|E(B)| \ge (1 - \alpha + \mu +o(1))^rm(1 \pm o(1))> m/2$ by the choice of $\alpha$, contradicting that $|E(B)| \le m/2$. Fact \ref{fc:cutsize} then implies that $|R| \ge |E(A, B)| \ge (\frac{1}{2} - \epsilon_r)m(1 \pm o(1))$.

Now assume that $B$ contains more than one connected component. Thus we can write $B = B_1 \cup B_2$, where $B_1$ and $B_2$ is an arbitrary bipartition of these components. Assume $|B_1| \ge |B_2|$. We show that $|A| \ge |B_1|$. Assume for a contradiction that this is not the case. Then $|A| < \min\{|A \cup B_2|, |B_1|\}$ and further we have $\max\{|A \cup B_2|, |B_1|\} \le |B_1 \cup B_2|$. This means that $A \cup B_2$ and $B_1$ give a more balanced bipartition, contradicting the minimality of $||A| - |B||$. Since $|A| + |B_1| + |B_2| = n$ and $|A| \ge |B_1| \ge |B_2|$, we have $|A| \ge n/3 \ge \alpha n$ (recall that $\alpha  = 1 - 2^{-1/r}$ and $\alpha \le 1/3$ when $r \ge 2$). Once again we can apply Fact \ref{fc:cutsize} to conclude that $|R| \ge |E(A, B)| \ge (\frac{1}{2} - \epsilon_r)m(1 \pm o(1))$.
\end{proof}

\section{A CSP-SAT Algorithm}

A $(d, k)$-CSP $\Psi$ is defined by a set of variables $X$ taking values in an alphabet $\Sigma$ of size $d$ and a set $\calC$ of constraints each on most $k$ of these variables. We write $\Psi = (X, \calC)$ to specify the variables and the constraint set. We will assume that the CSP is represented by the set of truth tables of its constraints. Observe that a $(2, k)$-CSP can be represented as a $k$-CNF. An assignment to the variables {\it satisfies} $\Psi$ if it satisfies every constraint. The variable frequency of $\Psi$ is the largest number of constraints that any variable appears in. Given a partial assignment $\rho$ which gives values to a set $Y \subseteq X$, the restriction of $\Psi$ by $\rho$ is denoted by $\Psi|_{\rho}$ which is a CSP on $X \setminus Y$ and each constraint is restricted by fixing the values of variables in $Y$ by $\rho$.

Let $\Psi = (X, \calC)$ be a CSP. We construct a hypergraph $H_{\Psi} = (V, E)$ as follows. We set $V = \calC$, that is every constraint is represented by a vertex in $H$. For every variable $x \in X$ we create a hyperedge $e_x:= \{C \in \calC : x \in C\}$, that is the set of constraints containing $x$ form a hyperedge.

\begin{prop}
\label{prp:csp-component}
Assume that $H_{\Psi}$ consists of connected components $H_1, \ldots, H_t$. Then $\Psi$ can be expressed as $\wedge_{i = 1}^t \Psi_i$ where $H_{\Psi_i} = H_i$ for each $1 \le i \le t$. 
\end{prop}
\begin{proof}
This is immediate once we observe that for any constraints $C$ and $C'$ which are represented in $H_i$ and $H_j$, respectively, for some $i \ne j$, $C$ and $C'$ do not have any variable in common.
\end{proof}

\begin{prop}
\label{prp:csp-restriction}
Let $\Psi = (X, \calC)$ be a CSP and $H_{\Psi} = (V, E)$ be the corresponding hypergraph. Let $\rho$ be a partial assignment which gives values to a set $Y \subseteq X$. Then $H_{\Psi|_{\rho}}$ is obtained by removing all $e_x$ with $x \in Y$ from $H_{\Psi}$. Furthermore, if $\Psi$ is unsatisfiable so is $\Psi|_{\rho}$ and consequently if $\Psi|_{\rho}$ breaks into $\wedge_{i = 1}^t \Psi_i$ as in Proposition \ref{prp:csp-component}, then at least one of $\Psi_i$s is unsatisfiable.
\end{prop}

\begin{proof}
After restricting some of the variables, those variables disappear and some constraints get simplified. But no new constraint is introduced and hence the hypergraph is obtained by removing the corresponding hyperedges. If a CSP is unsatisfiable, obviously it is unsatisfiable under any partial assignment. If an unsatisfiable CSP is decomposed into disjoint CSPs, at least one of these CSPs is unsatisfiable, since otherwise we can take a satisfying assignment from each part and since they are on disjoint sets of variables together they form a satisfying assignment of the whole CSP.
\end{proof}

We are now ready to describe our CSP-SAT algorithm.

\begin{algorithm}
\label{alg:csp-sat}
\caption{$\ralg(\Psi)$}

\begin{algorithmic}

\State Construct $H_{\Psi} = (V, E)$.

\State Construct a small balanced separator $R$ as in Theorem \ref{thm:sep} (either probabilistically or deterministically).

\ForAll{$\rho \in \Sigma^{R}$} 

\State Let $\Psi|_{\rho} = \wedge_{i = 1}^t \Psi_i$ as in Proposition \ref{prp:csp-restriction}.
\ForAll {$i \in [t]$} 
\State Exhaustively check if $\Psi_i$ is satisfiable.
\EndFor
\If{all $\Psi_i$'s are satisfiable} \Return satisfiable.
\EndIf
\EndFor\\
\Return unsatisfiable.
\end{algorithmic}

\end{algorithm}

\newcounter{tmp1}
\setcounter{tmp1}{\value{thm}}

\setcounter{thm}{\value{dpll-alg}}

\begin{thm}
\label{thm:csp-sat}
Let $r\ge 2$ be a fixed integer, $m, d, k\ge 1$ be integers such that $k \in o(\sqrt{m})$.
Let $\epsilon_r = (1 - 2^{-1/r})^r$.
Let $\Psi$ be an $m$-variable $(d, k)$-CSP with variable frequency at most~$r$.
$\ralg(\Psi)$ correctly decides the satisfiability of $\Psi$. Moreover, if $d\ge 3$ then it runs 
deterministically in time $d^{(1 - \epsilon_r)m + o(m)}$, if $d=2$ then
it runs in expected time $2^{(1 - \epsilon_r)m + o(m)}$ if we find $R$ randomly,
and in deterministic time $2^{(1 - 2\epsilon_r^2)m + o(m)}$ if we find $R$ deterministically.
\end{thm}

\begin{proof}
The correctness of the algorithm follows immediately from Proposition \ref{prp:csp-restriction}. In polynomial time we can construct $H_{\Psi} = (V, E)$. Observe that $H_{\Psi}$ has $m$ edges each of size at most $r$ and it has vertex degree at most $k$. By Theorem \ref{thm:sep}, we can find a balanced separator of size $(\frac{1}{2} - \epsilon_r)m + o(m)$ deterministically in time $2^{(1 - 2\epsilon_r^2)m+o(m)}$ or probabilistically in expected polynomial time in $m$. After having found the balanced separator $R$, there are at most $d^{(\frac{1}{2} - \epsilon_r)m + o(m)}$ runs of the for loop over $\rho$.  For each restriction $\rho$ we spend $|\Psi|^{O(1)}$ time to compute the decomposition of $\Psi$. Then for each of these parts we exhaustively check its satisfiability in time at most $d^{m/2}$. Since $\Psi$ breaks into at most $m := |V|$ parts the total running time after finding $R$ is at most $|\Psi|^{O(1)}d^{(\frac{1}{2} - \epsilon_r)m + o(m) + m/2} = |\Psi|^{O(1)}d^{(1 - \epsilon_r)m + o(m)}$. 
For $d\ge 3$, since $2^m < 3^{(1 - \epsilon_r)m}$, the total running time including finding the separator is
bounded by $|\Psi|^{O(1)}d^{(1 - \epsilon_r)m + o(m)}$. The claim follows by noting that $|\Psi|\le O(mrd^k) \le d^{o(m)}$.
For $d=2$, if we use the randomized procedure to find $R$, the total expected running time will be at most $|\Psi|^{O(1)}2^{(1 - \epsilon_r)m + o(m)}$, and if we run the deterministic procedure to find $R$, the total running time is at most $|\Psi|^{O(1)}(2^{(1 - 2\epsilon_r^2)m+o(m)} + 2^{(1 - \epsilon_r)m + o(m)}) \le |\Psi|^{O(1)}2^{(1 - 2\epsilon_r^2)m + o(m)}$. 
\end{proof}

\setcounter{thm}{\value{tmp1}}

\noindent{\bf Remark. }We can slightly modify the algorithm and instead of performing exhaustive search on the disjoint parts of the CSP we can make a recursive call to the algorithm. It is easy to verify that this improves the savings by a factor of two.

\begin{cor}
\label{thm:avg-csp-sat}
Let $r\ge 1$ be a fixed real, $m, d, k\ge 1$ be integers such that $k \in o(\sqrt{m})$.
Let $\epsilon_r = (1 - 2^{-1/r})^r$.
Let $\Psi$ be an $m$-variable $(d, k)$-CSP with average variable frequency at most $r$. 
$\ralg(\Psi)$ correctly decides the satisfiability of $\Psi$. Moreover, if $d\ge 3$ then it runs 
deterministically in time $d^{(1 - \epsilon_{\lceil 2r \rceil}/2)m + o(m)}$, if $d=2$ then
it runs in expected time $2^{(1 - \epsilon_{\lceil 2r \rceil}/2)m + o(m)}$ if we find $R$ randomly,
and in deterministic time $2^{(1 - \epsilon_{\lceil 2r \rceil}^2)m + o(m)}$ if we find $R$ deterministically.
\end{cor}

\begin{proof} Consider the case $d\ge 3$.
There are at most $m/2$ variables with frequency $\ge \lceil 2r \rceil$. For all possible settings
of those variables run the $\ralg(\Psi)$ algorithm on $\Psi$ restricted to that setting. In such a
restricted formula all variables have frequency at most $\lceil 2r \rceil$. By the previous
theorem the running time will be $d^{m/2} \cdot d^{(1 - \epsilon_{\lceil 2r \rceil}){m/2} + o(m)}$. The case for $d=2$ is analogous.
\end{proof}

\noindent{\bf Remark. }Note that our algorithm trivially also counts the number of satisfying assignments, hence \#CSP-SAT, and Max-CSP-SAT versions.

\section{Upper Bounds for Tree-like Resolution}

In this section we use our separator theorem to give non-trivial refutations of CNF representations of unsatisfiable $(2, k)$-CSPs with bounded variable frequency. Recall that a $(2, k)$-CSP with $n$ constraints can be represented by a $k$-CNF with $n2^k$ clauses. This class of CSPs includes the extensively studied Tseitin formulas which essentially encode that in a simple graph the number of odd degree vertices is even. Here we consider a more general definition for hypergraphs due to Pudl{\'a}k and Impagliazzo~\cite{DBLP:conf/soda/PudlakI00}.

\begin{defi}
Let $H = (V, E)$ be a hypergraph and let $\lambda : V \rightarrow \{0, 1\}$. The Tseitin formula on $H$, $T(H, \lambda)$, has a variable $x_e$ for every edge $e \in E$ and states that for every $v \in V$, $\oplus_{e \ni v} x_e \equiv \lambda(v)$. When $\oplus_{v \in V}\lambda(v) \equiv 1$ we call $\lambda$ an odd charge labeling. 
\end{defi}

When $\lambda$ is an odd charge labeling and each edge has even cardinality, $T(H, \lambda)$ is unsatisfiable. When $H$ has maximum degree $k$, $T(H, \lambda)$ is a $(2, k)$-CSP. From here on we will use $T(H, \lambda)$ to denote both the CSP and its CNF representation when it is clear from the context which one we refer to.

\begin{thm}
\label{thm:dpll}
Let $\Psi$ be an unsatisfiable $(2, k)$-CSP with variable frequency at most $r$ on $m$ variables. If $r$ is fixed and $k \in o(\sqrt{m})$, then there exists a tree-like resolution refutation of the CNF representation of $\Psi$ of size $2^{(1 - \epsilon_r)m + o(m)}$, where $\epsilon_r = (1 - 2^{-1/r})^r$. 
\end{thm}

\begin{proof}
Using Corollary \ref{cor:dtree} it is sufficient to give a decision tree for $\Psi$ of depth at most $(1 - \epsilon_r)m + o(m)$.

We will make use of Theorem \ref{thm:sep} applied to $H_{\Psi}$ and the strategy is quite immediate. The decision tree starts with a complete binary tree on all variables corresponding to the hyperedges in the balanced separator $R$ given by Theorem \ref{thm:sep} of size at most $(\frac{1}{2} - \epsilon_r)m + o(m)$. For every leaf with partial assignment $\rho$, by the separator property, Proposition \ref{prp:csp-restriction} and Proposition \ref{prp:csp-component} we can write $\Psi|_{\rho} = \wedge_{i = 1}^t\Psi_i$ for some $t$, where $\Psi_i$s are on disjoint sets of at most $m/2$ variables. Furthermore at least one of $\Psi_i$s is unsatisfiable. We then append this leaf by the complete binary tree on all the variables in $\Psi_i$. It is clear that at every leaf of this tree a contradiction is forced. The depth of this tree is at most $(\frac{1}{2} - \epsilon_r)m + o(m) + m/2 = (1 - \epsilon_r)m + o(m)$ and we are done.
\end{proof}

\begin{cor}
\label{cor:width}
Let $\Psi$ be an unsatisfiable $(2, k)$-CSP with variable frequency at most $r$ on $m$ variables. If $r$ is fixed and $k \in o(\sqrt{m})$, then there exists a resolution refutation of the CNF representation of $\Psi$ of width $(1 - \epsilon_r)m + o(m)$, where $\epsilon_r = (1 - 2^{-1/r})^r$. 
\end{cor}

\begin{proof}
This is immediate from Theorem \ref{thm:dpll} and the seminal result of Ben-Sasson and Wigderson~\cite{DBLP:journals/jacm/Ben-SassonW01} which states that if a $k$-CNF formula has a tree-like resolution refutation of size $2^w$, then it has a width $w + k$ resolution refutation.
\end{proof}

\begin{cor}
\label{cor:ts}
Let $H = (V, E)$ be an $r$-uniform hypergraph of maximum degree $k$ where $r$ is even and let $\lambda : V \rightarrow \{0, 1\}$ be an odd charge labeling. If $r$ is fixed and $k \in o(\sqrt{|E|})$ then $T(H, \lambda)$ can be refuted in tree-like resolution in size $2^{(1 - \epsilon_r)|E| + o(|E|)}$, where $\epsilon_r = (1 - 2^{-1/r})^r$.
\end{cor}

\begin{proof}
Observe that each variable appears in $r$ constraints.
\end{proof}

The case $r = 2$ corresponds to the the usual Tseitin tautologies on simple graphs. We give a finer analysis for this case which also involves a sharper derandomization.

\begin{thm}
\label{thm:tseitn-aut}

There exists a deterministic algorithm which on input $T(G, \lambda)$ where graph $G = (V, E)$ is of maximum degree $\Delta(G) \in o(\sqrt{|E|})$ and $\lambda$ is any odd charge labeling, produces a tree-like resolution refutation of $T(G, \lambda)$ in time $$|E|^{O(1)}2^{(1 - 2(1 - \frac{1}{\sqrt{2}})^2)|E| + o(|E|)} \le 2^{0.83|E|}.$$
\end{thm}

\begin{proof}

We will construct a decision tree for $T(G, \lambda)$ and apply Theorem \ref{thm:dtree}.

Let $n = |V|$, $m = |E|$, $k=\Delta(G)$. We will build the decision tree inductively. Each node of the decision tree
will be labeled by a pair $(G',\lambda')$, where $G'=(V',E')$, $V'\subseteq V$, $E'\subseteq E$ and $\lambda'$ is an odd charge labeling of $G'$. 
At each node of the decision tree, we query the value of some edge $e=(u,v)\in E'$, and depending on the value $b\in \{0,1\}$ of the variable corresponding to the edge $e$ we descend to a child labeled by $((V',E'\setminus \{e\}),\lambda'^{e,b})$. Here, $\lambda'^{e,b}$ equals $\lambda'$ except $\lambda'^{e,b}(v)=b \oplus \lambda'(v)$ and $\lambda'^{e,b}(u)=b \oplus \lambda'(u)$. Hence, if  $\lambda'$ is an odd charge labeling so is $\lambda'^{e,b}$. (In some nodes we will reduce $G'$ even more as explained below.)
The root of the decision tree is labeled by $(G,\lambda)$.

A node of the decision tree labeled by $(G',\lambda')$ will be a leaf  
if there is a degree zero vertex $v$ in $G'$ with $\lambda'(v)=1$. In this case the constraint corresponding to $v$ is falsified, and so some CNF clause associated to that constraint is also falsified. 
If $G'$ contains a vertex $v$ of degree $1$, we query the value $b$ of the incident edge $e=(u,v)\in E'$ and we descend to children labeled by $((V',E\setminus \{e\}),\lambda'^{e,b})$. Notice, $((V',E'\setminus \{e\}),\lambda'^{e,b})$ violates the parity constraint at $v$ either for $b=0$ or $b=1$. So one of the children is a leaf.

If $G'$ contains a vertex $v$ of degree $2$, we query the value of one of the incident edges $e=(u,v)\in E'$, and the two children are labeled by $((V',E'\setminus \{e\}),\lambda'^{e,b})$. If $G'$ does not contain any vertex of degree less than 3,
we find a balanced separator $R$ in $G'$ of size at most $(\frac{1}{2} - \epsilon)|E'|+C k \sqrt{|E'|}$, where $\epsilon = (1 - \frac{1}{\sqrt{2}})^2$ and $C$ is a constant. (To find the separator we use an algorithm described later. The separator exists by a remark after Theorem \ref{thm:sep}.) We query all the edges $e \in R$ one by one and after each query we descend from $((V',E'),\lambda')$
to $((V',E'\setminus \{e\},\lambda'^{e,b})$ via an edge labeled by $b$. After asking the last edge $e$ of $R$, instead of descending to $((V',E'\setminus \{e\}),\lambda'^{e,b})$ we descend to $((V'',E''),\lambda'')$ where $E''\subseteq (E'\setminus \{e\})$ forms a connected component on which $\lambda'^{e,b}$
is an odd charge labeling. We let $V''$ be the vertices belonging to that component and we set $\lambda''$ to $\lambda'^{e,b}$ restricted to $V''$. We repeat the whole process until $|E'|\le k\sqrt{m}$ at which point we query all the remaining edges one by one to identify a vertex violating its parity constraint.

To find the balanced separator we proceed as follows: If $G'$ has minimum degree at least 3 we have $|V'| \le 2|E'|/3$. As every balanced separator corresponds to a cut in $G'$, we look for a small balanced separator by exhaustively checking all cuts $S \subseteq V'$ whether they define a small balanced separator. This will be done in time $|E'|^{O(1)} 2^{2|E'|/3}$.

We prove by induction on the height of a node that the subtree of a node labeled by $((V',E'),\lambda')$ has at most $2^{(1 - 2\epsilon)|E'|+4Ck\sqrt{|E'|}+k\sqrt{m}}$ leaves. 
If the node is a leaf of the decision tree, the bound holds. If $|E'| \le k\sqrt{m}$ then the bound holds as well. If $(V',E')$ contains a vertex of degree 1,
then one of its children is a leaf and by the induction on the non-leaf descendant we have: 
\begin{eqnarray*}
1+2^{(1 - 2\epsilon)(|E'|-1)+4Ck\sqrt{|E'|-1}+k\sqrt{m}} &\le& 2^{(1 - 2\epsilon)} \cdot 2^{(1 - 2\epsilon)(|E'|-1)+4Ck\sqrt{|E'|-1}+k\sqrt{m}}\\
&\le& 2^{(1 - 2\epsilon)|E'|+4Ck\sqrt{|E'|}+k\sqrt{m}}.
\end{eqnarray*}

If $(V',E')$ contains a vertex of degree 2, then its
children are labeled by sub-graphs containing a degree one vertex, so the number of leaves in the entire sub-tree is bounded by induction by: $$2 + 2\cdot 2^{(1 - 2\epsilon)(|E'|-2)+4Ck\sqrt{|E'|-2}+k\sqrt{m}}
\le 2^{(1 - 2\epsilon)|E'|+4Ck\sqrt{|E'|}+k\sqrt{m}}$$ as desired. If $(V',E')$ has minimum degree at least three, then we find a separator $R$ of size at most $(\frac{1}{2} - \epsilon)|E'|+C k \sqrt{|E'|}$
and we descend to $2^{|R|}$ nodes labeled by sub-graphs on at most $|E'|/2$ edges. 
By induction the size of the subtree will be at most

\begin{eqnarray*} 
2^{|R|} \cdot 2^{(1 - 2\epsilon)|E'|/2+4Ck\sqrt{|E'|/2}+k\sqrt{m}}  &\le& 2^{(1 - 2\epsilon)|E'|/2+C k \sqrt{|E'|}} \cdot 2^{(1 - 2\epsilon)|E'|/2+4C k\sqrt{|E'|/2}+k\sqrt{m}} \\
&\le&  2^{(1 - 2\epsilon)|E'|+4C k\sqrt{|E'|}+k\sqrt{m}}
\end{eqnarray*}
as required.

Hence, for $k\in o(\sqrt{m})$, the size of the decision tree is $2^{(1 - 2(1 - \frac{1}{\sqrt{2}})^2)|E| + o(|E|)}$. Descendants of each node of the tree can be determined in time $|E|^{O(1)}$ given the balanced separators so to bound the overall running time we only need to bound the time spent in search for the balanced separators. If we process the graphs $(V',E')$ labeling the decision tree always in some canonical way (we pick the edge to query irrespective of $\lambda'$), then we get at most $n$ distinct graphs at a given level of the decision tree, each on a disjoint set of vertices. Hence, there will be at most $|V|\cdot |E|$ distinct graphs in total. For each of them we need to call the balanced separator procedure only once. Hence, the time spent in them will be at most $|E|^{O(1)} \cdot 2^{2|E|/3}$. As $(1 - 2(1 - \frac{1}{\sqrt{2}})^2) > 2/3$ the claim follows.
\end{proof}

\section{Conclusion}

We showed that we can remove a small number of edges from $r$-uniform hypergraphs with bounded vertex degree to break them into connected components each with at most half of the edges. This was used to solve the satisfiability of sparse CSPs. Many problems can be stated in terms of hypergraphs. Therefore we hope that our separator theorem will find more applications.

Furthermore we believe that the following problems deserve attention:

\begin{enumerate}
\item Give $(d, k)$-SAT algorithms with savings not depending on $d$, or show that under some complexity assumption this is not possible.
\item Extend our result on Tseitin formulas to larger classes of $k$-CNFs which admit resolution proofs with non-trivial savings.
\item Find deterministic algorithms which produce resolution proofs with non-trivial savings.
\end{enumerate}

\section{Acknowledgments}

The research of Kouck{\'y} and Talebanfard was supported by GA{\v C}R grant 19-27871X. R{\"o}dl's research was supported by NSF grant DMS 1764385. We are grateful to anonymous reviewers whose comments allowed us to improve the presentation of the paper.

\bibliographystyle{alpha} 
\bibliography{ref}

\end{document}